\documentclass{llncs}
\usepackage{graphicx}
\usepackage{amsfonts}
\usepackage{amsmath}
\usepackage{hyperref}
\title{Integrity-protecting block cipher modes --- Untangling a tangled web}
\author{Chris J. Mitchell}
\institute{Royal Holloway, University of London, Egham TW20 0EX, UK \email{c.mitchell@rhul.ac.uk};~~ \url{www.chrismitchell.net}}
\begin{document}

\maketitle

\begin{abstract}
This paper re-examines the security of three related block cipher modes of operation designed to
provide authenticated encryption. These modes, known as PES-PCBC, IOBC and EPBC, were all proposed in the mid-1990s. However, analyses of security of the latter two modes were published more recently.  In each case one or more papers describing security issues with the schemes were eventually published, although a flaw in one of these analyses (of EPBC) was subsequently discovered --- this means that until now EPBC had no known major issues. This paper establishes that, despite this, all three schemes possess defects which should prevent their use --- especially as there are a number of efficient alternative schemes possessing proofs of security.

\keywords{authenticated encryption \and cryptanalysis \and mode of operation}
\end{abstract}

\section{Introduction} \label{s-Intro}

This paper is a somewhat tangled story of three different, albeit very closely related, proposals for a block cipher mode of operation providing authenticated encryption. Sadly, all of the schemes been
shown to be insecure --- often in quite different ways. This is, to the author's knowledge, the
first paper to bring together the three strands of research; at the same time errors in
previous cryptanalysis are acknowledged and further attacks described.

All three of the schemes we examine are examples of a `special' mode of operation for block ciphers,
designed to offer `low cost' combined integrity and confidentiality protection by combining
encryption with the addition of simple (or fixed) redundancy to the plaintext. The underlying idea
is to design the mode so that modifying the ciphertext without invalidating the added redundancy is
impossible without knowledge of the encryption key. Such modes are the theme of section 9.6.5 of
Menezes, van Oorschot and Vanstone's landmark book \cite{Menezes97}. Two main methods for adding
redundancy have been proposed:
\begin{itemize}
\item add a fixed block (or blocks) at the end of the plaintext, which may be public or secret (in the latter case the block acts as an auxiliary key);
\item append to the plaintext some easily computed and simple (public) function of the
    plaintext.
 \end{itemize}
In either case we refer to the block added to the end of the plaintext as a \emph{Manipulation
Detection Code (MDC)}\@. Whichever approach is adopted, the method for computing the MDC needs to
be simple, or it offers no advantage over the more conventional `encrypt then MAC' approach. In all
the modes we examine, the MDV (also known as an \emph{Integrity Control Value} (ICV)) is defined to
be a fixed, possibly secret, final plaintext block.

The first of the three schemes we examine is known as PES-PCBC\@.  The name derives from it being a
version of PCBC mode designed specifically for use in a Privacy Enhanced Socket (PES) protocol.
There are actually a number of modes known as PCBC (Plaintext-Ciphertext Block Chaining); the
version on which PES-PCBC was based is the one incorporated in Kerberos version 4. This mode was
shown to be insecure for the purposes of integrity protection by Kohl, \cite{Kohl90}. For a
discussion of the weaknesses of other variants of PCBC see \cite{Mitchell05}. The design goal for
PES-PCBC was to enhance the security of PCBC by preventing the known attacks. This scheme and its
properties are discussed in Section~\ref{s-PES-PCBC}.

The second scheme, known as IOBC (short for \emph{Input and Output Block Chaining}) was published
in 1996 by Rechacha \cite{Recacha96}. IOBC is a straightforward variant of PES-PCBC. The paper
describing IOBC was originally published in Spanish, and it wasn't until an English language
translation was kindly provided by the author in around 2013\footnote{See
\url{https://inputoutputblockchaining.blogspot.com/}} that any further discussion of the scheme
appeared. This scheme and its properties are discussed in Section~\ref{s-IOBC}.

The last of the three schemes, known as EPBC (short for \emph{Efficient     Error-Propagating Block
Chaining}) was published in 1997 by Z\'{u}quete and Gudes \cite{Zuquete97}. It is very similar to
IOBC, and is designed to be used in exactly the same way. The design goal was to address issues in
IOBC which restricted its use to relatively short messages. A possible method of cryptanalysis
allowing certificational forgeries was published in 2007 \cite{Mitchell07}, although Di et al.\
\cite{Di15} showed in 2015 that the attack does not work as claimed. The scheme and its level of
security are discussed in Section~\ref{s-EPBC}.

In Section~\ref{s-other} we examine other more general attacks which apply to all, or at least
large classes of, schemes sharing the same underlying structure as PES-PCBC, IOBC and EPBC\@. In
doing so we establish that all three modes suffer from attacks, and so they should not be adopted.
We conclude the paper in Sections~\ref{s-othermodes} and \ref{s-conclusions} by first briefly
mentioning two other related modes (and their analyses), and then summarising the main conclusions
that can be drawn from the analyses given.

\section{Preliminaries}

We start by introducing some notation and assumptions. All three modes operate using a block cipher. We write:
\begin{itemize}
\item $n$ for the plaintext/ciphertext block length of this cipher;
\item $e_K(P)$ for the result of block cipher encrypting the $n$-bit block $P$ using the secret key $K$;
\item $d_K(C)$ for the result of block cipher decrypting the $n$-bit block $C$ using the key $K$; and
\item $\oplus$ for bit-wise exclusive-or.
\end{itemize}
Finally we suppose the plaintext to be protected using the mode of operation is divided into a sequence of $n$-bit blocks (if necessary, having first been padded): $P_1,P_2,\ldots,P_t$, where $P_t$ is equal to the MDC.

\section{PES-PCBC}  \label{s-PES-PCBC}

The description follows Z\'{u}quete and Guedes \cite{Zuquete96}, although we use the notation of
\cite{Zuquete97}. The scheme uses two secret $n$-bit Initialisation Vectors (IVs), denoted by $F_0$
and $G_0$.  The nature of their intended use is not described in \cite{Zuquete96}; however it is
stated in \cite{Zuquete97} that the `initial values of $F_{i-1}$ and $G_{i-1}$ are distinct, secret
initialisation vectors', which is what we assume below.

\subsection{PES-PCBC operation}

The PES-PCBC encryption of the plaintext $P_1,P_2,\ldots,P_t$ operates as follows:
\begin{eqnarray*}
G_i & = & P_i \oplus F_{i-1},\:\:\:(1\leq i\leq t), \label{PES-GPF} \\
F_i & = & e_K(G_i),\:\:\:(1\leq i\leq t), \\
C_i & = & F_i \oplus G_{i-1},\:\:\:(1\leq i\leq t). \label{PES-CFG}
\end{eqnarray*}
The resulting ciphertext is $C_1,C_2,\ldots,C_t$.

The operation of the mode (when used for encryption) is shown in Figure~\ref{fig:PES-PCBC}. Note that we refer to the values $F_i$ and $G_i$ as `internal' values, as they are computed during encryption, but they do not constitute part of the ciphertext.

\begin{figure}
\resizebox{\textwidth}{!}
{\includegraphics*[0.5cm,1.7cm][19.5cm,13.5cm]{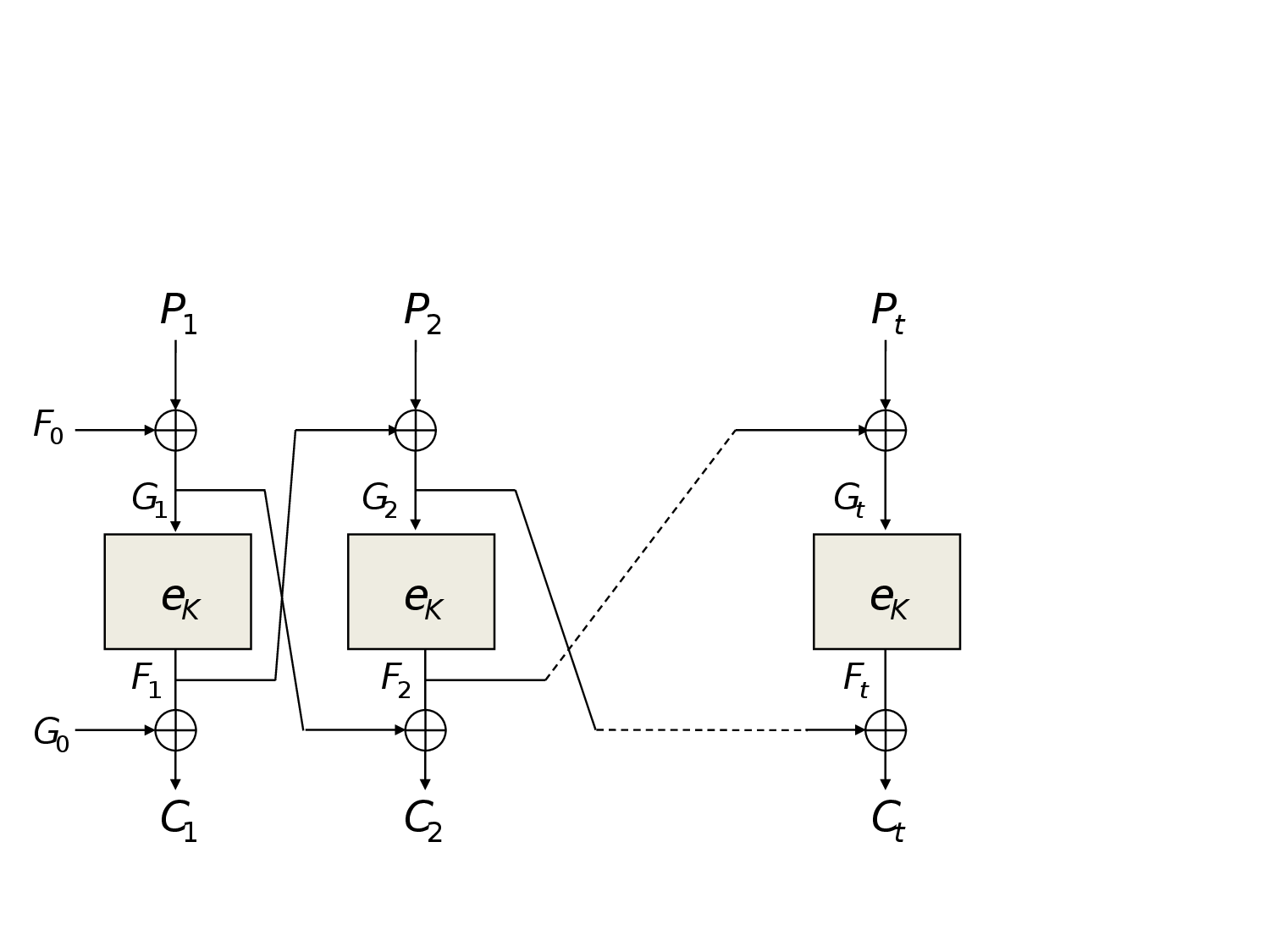}}
\caption{PES-PCBC encryption}  \label{fig:PES-PCBC}
\end{figure}

\subsection{Cryptanalysis}  \label{ss-PES-PCBC-cryptanalysis}

It is clear that weaknesses in PES-PCBC when applied for integrity protection were discovered soon
after its publication in 1996.  The 1997 paper by Z\'{u}quete and Guedes, \cite{Zuquete97}, briefly
outlines a known-plaintext attack allowing simple forgeries. We next describe a slightly simplified
variant of this attack, requiring just one known plaintext block instead of two.

First observe that, using the same notation as before, PES-PCBC decryption operates as follows:
\begin{eqnarray*}
F_i & = & C_i \oplus G_{i-1},\:\:\:(1\leq i\leq t), \label{PES-DCFG} \\
G_i & = & d_K(F_i),\:\:\:(1\leq i\leq t), \\
P_i & = & G_i \oplus F_{i-1},\:\:\:(1\leq i\leq t), \label{PES-DGPF},
\end{eqnarray*}
and the receiver of an encrypted message will accept it as genuine if the final recovered plaintext
block $P_t$ equals the expected MDC.

The following result captures the attack.

\begin{theorem}  \label{t-PES-PCBC}
Suppose the ciphertext $C_1,C_2,\ldots,C_t$ was constructed using PES-PCBC from the plaintext
$P_1,P_2,\ldots,P_t$, and that $j$ satisfies $1<j<t$. Suppose the $(t+2)$-block ciphertext
$C'_1,C'_2,\ldots,C'_{t+2}$ is constructed as follows:
\begin{eqnarray*}
C'_i & = & C_i,\:\:\:(1\leq i\leq j), \label{PES-fake1} \\
C'_{j+1} & = & P_j, \\
C'_i & = & C_{i-2},\:\:\:(j+2\leq i\leq t+2). \label{PES-fake2}
\end{eqnarray*}
When decrypted to yield $P'_1,P'_2,\ldots,P'_{t+2}$, the value of the final plaintext block
$P'_{t+2}$ will equal $P_t$ for the original (untampered) message, and hence will pass the
integrity check.
\end{theorem}

\begin{proof}
In the discussion below we refer to the `internal values' generated during decryption of
$C'_1,C'_2,\ldots,C'_{t+2}$ as $F'_i$ and $G'_i$. First note that, trivially:
$F'_i = F_i,~~G'_i = G_i$, and $P'_i = P_i$, ($1\leq i\leq j$). Next observe that
\begin{eqnarray*}
F'_{j+1} & = & C'_{j+1}\oplus G'_j \\
& = & P_j \oplus G_j \\
& = & F_{j-1}.
\end{eqnarray*}
Hence $G'_{j+1} = d_K(F'_{j+1}) = d_K(F_{j-1}) = G_{j-1}$. Finally, we have
$P'_{j+1} = G'_{j+1}\oplus F'_j = G_{j-1}\oplus F_j = C_j$. Since $C'_{j+2} = C_j$, $F'_{j+1}=F_{j-1}$ and $G'_{j+1} = G_{j-1}$, it is immediate that
$F'_{j+2}=F_j$, $G'_{j+2} = G_j$, and $P'_{j+2} = P_j$, and the desired result
follows. \qed
\end{proof}

\subsection{Impact}

The above attack shows that, given just one PES-PCBC-encrypted message and knowledge of only a
single plaintext block for this encrypted message, a `fake' message can be constructed that will be
guaranteed to pass the integrity checks. This fact meant that it has been recognised since 1996/97
that PES-PCBC should not be used.

Before proceeding note that the originally proposed context of use for PES-PCBC, as described in \cite{Zuquete96}, involved including an encoded version of the message length in the first plaintext block. In such a case the attack described in Theorem~\ref{t-PES-PCBC} will not work since it involves lengthening the message by two blocks. However, a slightly more involved version of the Theorem~\ref{t-PES-PCBC} attack (outlined in \cite{Zuquete97}) avoids changing the message length and hence works even if the message length is encoded in the plaintext --- at the cost of requiring knowledge of two plaintext blocks instead of one.

\section{IOBC}  \label{s-IOBC}

The IOBC mode was published in 1996 by Recacha \cite{Recacha96}, the same year as the publication
of PES-PCBC\@.  One might reasonably conclude that the design of IOBC, as a modification to
PES-PCBC, was motivated by the weaknesses in PES-PCBC, but curiously the Recacha paper does not
even mention PES-PCBC\@. Certainly, the inclusion of the function $g$ in the feedback stops the
attack on PES-PCBC working --- at least in a naive way.

\subsection{IOBC operation}  \label{ss-IOBC-operation}

We start by describing the operation of the IOBC mode of operation. We base the description on
Recacha's 1996 paper \cite{Recacha96}, although we use the same notation as in the description of
PES-PCBC\@.  We suppose that the cipher block length $n$ is a multiple of four (as is the case for
almost all commonly used schemes), and put $n=2m$ where $m$ is even. The scheme uses two secret $n$-bit IVs, denoted by $F_0$ and $G_0$.  The nature of the intended restrictions on their use is not altogether clear; one suggestion in the original Recacha paper \cite{Recacha96} is that they should be generated as follows.

Suppose $K'$ is an auxiliary key used solely for generating the IVs.  Suppose also that $S$ is a
sequence number, managed so that different values are used for every message.  Then $F_0=e_{K'}(S)$
and $G_0=e_{K'}(F_0)$.  For the purposes of this paper we assume that $F_0$ and $G_0$ are always
generated this way, and thus the scheme can be thought of as employing a pair of block cipher keys
and a non-secret, non-repeating, sequence number (which must be carefully managed to prevent
accidental re-use of sequence number values).  Note that special measures will need to be taken if
the same key is to be used to encrypt communications in both directions between a pair of parties.
Avoiding sequence number re-use in such a case could be achieved by requiring one party to start
the sequence number they use for encryption at a large value, perhaps halfway through the range.

The IOBC encryption of the plaintext $P_1,P_2,\ldots,P_t$ operates as follows:
\begin{eqnarray*}
G_i & = & P_i \oplus F_{i-1},\:\:\:(1\leq i\leq t), \label{GPF} \\
F_i & = & e_K(G_i),\:\:\:(1\leq i\leq t), \\
C_i & = & F_i \oplus g(G_{i-1}),\:\:\:(2\leq i\leq t), \label{CFG}
\end{eqnarray*}
where $C_1 = F_1\oplus G_0$ and $g$ is a function that maps an $n$-bit block to an $n$-bit block,
defined below.  The operation of the mode (when used for encryption) is shown in
Figure~\ref{fig:mode}.

\begin{figure}
\resizebox{\textwidth}{!}
{\includegraphics*[0.5cm,1.7cm][19.5cm,13.5cm]{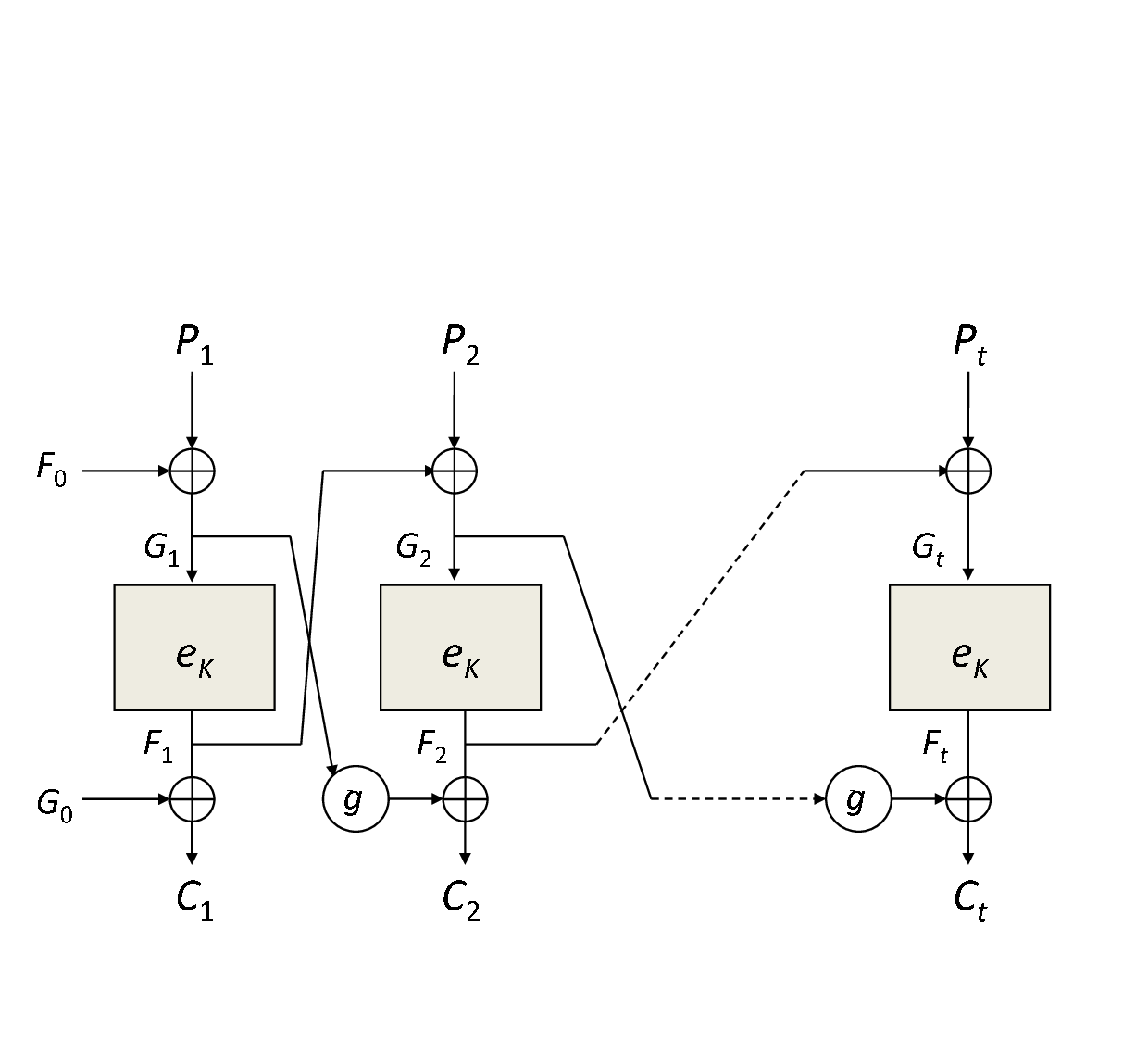}}
\caption{IOBC encryption}  \label{fig:mode}
\end{figure}

The function $g$ is defined as follows.  Suppose $X$ is an $n$-bit block, where $n=2m$.  Suppose
also that $X=L||R$ where $L$ is the leftmost $m-1$ bits of $X$ and $R$ is the rightmost
$m+1$ bits of $X$ (and, as throughout, $||$ denotes concatenation). Then
\[ g(X) = (>_1(L))||(>_1(R)) \]
where $>_i$ denotes a rightwards (cyclic) shift by $i$ bit positions.

Decryption operates similarly.  We have:
\begin{eqnarray*}
F_i & = & C_i \oplus g(G_{i-1}),\:\:\:(2\leq i\leq t), \\
G_i & = & d_K(F_i),\:\:\:(1\leq i\leq t), \\
P_i & = & G_i \oplus F_{i-1},\:\:\:(1\leq i\leq t).
\end{eqnarray*}
and $F_1 = C_1\oplus G_0$, where $d$ denotes block cipher decryption.

It should be clear that PES-PCBC is the same as IOBC with the exception that in PES-PCBC the function $g$ is the identity function.

\subsection{Remarks on use}  \label{ss-IOBC-remarks}

As described above, we assume throughout that the IVs $F_0$ and $G_0$ are derived by
ECB-mode-encrypting a sequence number using a secondary key. Thus the ciphertext blocks will be a
function of this serial number (as well as the pair of keys used).  We thus write $[S], C_1, C_2,
\ldots, C_t$ for a sequence of ciphertext blocks, meaning that $C_1, C_2, \ldots, C_t$ were
encrypted using the sequence number $S$.  This is logical, since the sequence number will need to
be sent or stored with the ciphertext to enable correct decryption.

IOBC should only be used with relatively short messages.  As specified by Recacha \cite{Recacha96}
(and for reasons which become clear below), a message to be encrypted using IOBC shall contain at
most $n^2/2-1$ plaintext blocks, where $n$ is the plaintext block length. Thus for $n=64$ and
$n=128$, the two most commonly used block lengths, a message shall contain at most 2047 and 8191
blocks, respectively.

As with all modes we discuss here, data integrity is achieved by adding an MDC to the end of the
plaintext.

\subsection{Cryptanalysis}  \label{ss-IOBC-cryptanalysis}

We start by giving a simple result that is implicit in Recacha \cite{Recacha96}. It is interesting
to note that this result applies regardless of the choice of the choice of function $g$, i.e.\ to
any mode operating as in Figure~\ref{fig:mode}.

\begin{lemma}[Mitchell, \cite{Mitchell13}]  \label{splicing-lemma}
Suppose $[S], C_1, C_2, \ldots, C_t$ and $[S'], C'_1, C'_2, \ldots, C'_{t'}$ are IOBC encrypted
versions of the plaintext sequences $P_1, P_2, \ldots, P_t$ and $P'_1, P'_2, \ldots, P'_{t'}$,
respectively.  If the ciphertext:
\begin{eqnarray*}
& [S'], C^*_1, C^*_2, \ldots, C^*_{t-v+u} = \\
& [S'], C'_1, C'_2, \ldots, C'_{u-1}, C_v\oplus g(G'_{u-1})\oplus
g(G_{v-1}), C_{v+1}, \ldots, C_t
\end{eqnarray*}
is submitted for IOBC decryption (where $1<u<t'$ and $1<v<t$, and $G_{v-1}$ and $G'_{u-1}$ are values computed during the encryption of the respective sequences of blocks), then the resulting sequence
of plaintext blocks $P^*_1, P^*_2, \ldots, P^*_{t-v+u}$ will be equal to
\[ P'_1, P'_2, \ldots, P'_{u-1}, P_v \oplus F'_{u-1} \oplus F_{v-1}, P_{v+1}, P_{v+2}, \ldots, P_t.
\]
\end{lemma}

Lemma~\ref{splicing-lemma} suggests a way of forging an IOBC-encrypted message so that the final
block will contain the correct MDC\@.  However, the problem remains of discovering
$g(G'_{u-1})\oplus g(G_{v-1})$, as this is used in constructing the forged ciphertext in the
statement of the lemma. Recacha \cite{Recacha96} discusses this very point, and explains that
making this difficult motivates the inclusion of the function $g$ in the design of IOBC --- that
is, if $g$ was not included (as is the case for PES-PCBC), then simple forgeries could be achieved.

We also have the following, also implicit in Recacha's 1996 paper.

\begin{lemma}[Mitchell, \cite{Mitchell13}]  \label{induction}
Suppose $[S], C_1, C_2, \ldots, C_t$ is the encryption of $P_1, P_2, \ldots, P_t$ using IOBC, and
that $F_i$ and $G_i$ are as defined in Section~\ref{ss-IOBC-operation}.  Then:
\begin{enumerate}
\item[(i)] $C_{j+1}\oplus P_{j+2}=g(G_j)\oplus G_{j+2}$,~~$1\leq j\leq t-2$;
\item[(ii)] $\bigoplus_{i=1}^k g^{k-i}(C_{j+2i-1}\oplus P_{j+2i}) = g^k(G_j) \oplus
    G_{j+2k}$,~~$1\leq j\leq t-2$,~$1\leq k\leq(t-j)/2$.
\end{enumerate}
\end{lemma}

It is not hard to see that if $g^k(G_j)=G_j$ for some $k$, then Lemma~\ref{induction}(ii) could be
combined with Lemma~\ref{splicing-lemma} to yield a forgery attack (given a ciphertext message with
corresponding known plaintext).  This point is made by Recacha \cite{Recacha96}, who explains that
the bit permutation $g$ has been chosen so that the smallest integer $i>1$ such that $g^i$ is the
identity permutation is $(n/2-1)(n/2+1)=n^2/4-1$ (since $m=n/2$ is even).  The restriction that the maximum length of
messages encrypted using IOBC is $n^2/2-1$, as defined in Section~\ref{ss-IOBC-remarks}, prevents
this problem arising in practice.

However, in some cases $g^k$ is `close' to the identity permutation for somewhat smaller values of
$k$.  The following result highlights this for two practically important values of $n$.   Observe
that analogous results can be achieved for any $n$.

\begin{lemma}[Mitchell, \cite{Mitchell13}]  \label{lemma64_128}
Suppose $X$ is a randomly selected $n$-bit block.
\begin{enumerate}
\item[(i)] If $n=64$ then $\Pr(X=g^{341}(X))=2^{-22}$; and
\item[(ii)] if $n=128$ then $\Pr(X=g^{1365}(X))=2^{-42}$.
\end{enumerate}
\end{lemma}

The above result can now be used in a straightforward way to enable message forgeries. As described
in detail in \cite{Mitchell13}, if $n=64$, given an IOBC ciphertext containing at least 685 blocks
and some of the corresponding plaintext, then it is possible to create a forged ciphertext that
will pass integrity checks with probability $2^{-22}$. 685 is clearly much less than the defined
message length limit of 2047 blocks. A precisely analogous attack works for $n=128$, although the
success probability is only $2^{-42}$.

\subsection{Impact}

The attack outlined immediately above could be prevented by further curtailing the maximum length
for messages, but this would in turn further limit the applicability of the scheme. Moreover, the
lack of a formal proof of security means that other attacks are possible.  Indeed, a simple chosen
plaintext forgery attack was outlined in \cite{Mitchell13}, although it requires approaching
$2^{n/2}$ ciphertexts for chosen plaintexts (this attack is discussed further in
Section~\ref{ss-EPBC-chosen} below). These points strongly argue against adoption of this scheme.

\section{EPBC}  \label{s-EPBC}

The EPBC scheme was proposed by Z\'{u}quete and Guedes \cite{Zuquete97} the year after the
publication of IOBC\@. The primary design goal was to remove the message length restriction
inherent in the design of IOBC; it also enables a small efficiency improvement.  It was further
claimed by its authors to be more secure than IOBC.

\subsection{EPBC operation}  \label{ss-EPBC-operation}

The scheme operates in a very similar way to IOBC, exactly as shown in Figure~\ref{fig:mode}, and (like IOBC) requires that $n$ is even. The
only significant difference is in the choice of the function $g$.  The function $g$ for EPBC is not
bijective, unlike in IOBC, and operates as follows. Suppose $X$ is an $n$-bit block, where $X=L||R$
and $L$ and $R$ are $m$-bit blocks (and, as throughout, $||$ denotes concatenation). Then
\[ g(X) = (L \vee \overline{R})||(L \wedge \overline{R})
\]
where $\vee$ denotes the bit-wise inclusive or operation, $\wedge$ denotes the bit-wise logical and
operation, and $\overline{X}$ denotes the logical negation of $X$ (i.e.\ changing every zero to a
one and vice versa).

Much like with PES-PCBC, the IVs $F_0$, $G_0$ are required to be distinct, secret initialisation values.

\subsection{A flawed cryptanalysis}  \label{ss-EPBC-cryptanalysis}

We first give some simple results on $g$.

\begin{lemma}[Mitchell, \cite{Mitchell07}]  \label{g_output}
Suppose $g(X)=L'||R'$, where $X$ is an $n$-bit block and we let
$L'=(\ell'_1,\ell'_2,\ldots,\ell'_m)$ and $R'=(r'_1,r'_2,\ldots,r'_m)$ be $m$-bit blocks.  Then,
for every $i$ ($1\leq i\leq m$), if $\ell'_i=0$ then $r'_i=0$.
\end{lemma}

The above Lemma implies that output bit pairs $(\ell'_i,r'_i)$ can never be equal to (0,1).  In
fact, we can obtain the following more general result which gives Lemma~\ref{g_output} as a special
case.

\begin{lemma}[Mitchell, \cite{Mitchell07}]  \label{g_input}
Suppose that, as above, $X=L||R$ where $L=(\ell_1,\ell_2,\ldots,\ell_m)$ and
$R=(r_1,r_2,\ldots,r_m)$. Suppose also that $g(X)=L'||R'$ where
$L'=(\ell'_1,\ell'_2,\ldots,\ell'_m)$ and $R'=(r'_1,r'_2,\ldots,r'_m)$.  Then if $(\ell_i,r_i)\in
A$ then $(\ell'_i,r'_i)\in B$, where all possibilities for $A$ and $B$ are given in
Table~\ref{inout}.  Note that, for simplicity, in this table we write $xy$ instad of $(x,y)$.

\begin{table}
\caption{Input/output possibilities for $g$}  \label{inout}
\begin{center}\begin{tabular}[htb]{|c|c|}
\hline $A$ (set of input pairs) & $B$ (set of output pairs) \\
 \hline
 $\{00,01,10,11\}$ & $\{00,10,11\}$ \\
 \hline
 $\{01,10,11\}$ & $\{00,10,11\}$ \\
 $\{00,10,11\}$ & $\{10,11\}$ \\
 $\{00,01,11\}$ & $\{00,10\}$ \\
 $\{00,01,10\}$ & $\{00,10,11\}$ \\
 \hline
 $\{10,11\}$ & $\{10,11\}$ \\
 $\{01,11\}$ & $\{00,10\}$ \\
 $\{01,10\}$ & $\{00,11\}$ \\
 $\{00,11\}$ & $\{10\}$ \\
 $\{00,10\}$ & $\{10,11\}$ \\
 $\{00,01\}$ & $\{00,10\}$ \\
 \hline
 $\{11\}$ & $\{10\}$ \\
 $\{10\}$ & $\{11\}$ \\
 $\{01\}$ & $\{00\}$ \\
 $\{00\}$ & $\{10\}$ \\
 \hline
\end{tabular}
\end{center}
\end{table}
\end{lemma}

We next summarise the key part of the known-plaintext attack described in \cite{Mitchell07}. The
primary objective is to use knowledge of known plaintext/ciphertext pairs ($P_i$, $C_i$) to learn
the values of corresponding `internal pairs' ($F_i$, $G_i$). These can then be used in a fairly
straightforward way (as detailed in \cite{Mitchell07}) to construct a forged ciphertext which will
pass the integrity checks.

It is claimed in \cite{Mitchell07} that, assuming that we have sufficiently many known plaintext and ciphertext pairs, for sufficiently large $w$ there will only be one possibility for $F_{j+2w}$. Using knowledge of $P_{j+2w+1}$, this immediately gives certain knowledge of $G_{j+2w+1}$. I.e., for all sufficiently large values of $w$, complete knowledge can be obtained of $F_{j+2w}$ and $G_{j+2w+1}$.

However, more recently, Di et al.\ \cite{Di15} pointed out that the above analysis has a major
flaw.  The issue arises in the argument that, since $G_{j+1} = P_{j+1}\oplus F_j$, information
about forbidden bit pairs in $F_j$, combined with knowledge of $P_{j+1}$, gives information about
forbidden bit pairs in $G_{j+1}$.  Di et al.\ \cite{Di15} point out that if there are two
possibilities for a bit pair in $F_j$ then there will always still be two possibilities for the
corresponding bit pair in $g(G_{j+1})$ --- as opposed to the analysis in \cite{Mitchell07} which
suggests that the number of possibilities will be reduced to one with probability 1/6. That is, the
number of possibilities for a bit pair in $g(G_{j+1})$ will never go below two, preventing the
attack strategy working.

\subsection{Impact}

Di et al.\ \cite{Di15} were not able to suggest any further attacks apart from a brute force approach.  This suggests that EPBC may, after all, be secure.  However, in the next section we show otherwise.

\section{Other attacks}  \label{s-other}

We now consider other possible attacks. Given that all three modes we have considered share the
same underlying structure, as shown in Figure~\ref{fig:mode}, we focus on attacks that apply to
large classes of possibilities for the function $g$.

\subsection{A chosen plaintext forgery attack}  \label{ss-EPBC-chosen}

We start by giving an attack which will work for any function $g$, using a method outlined in
\cite{Mitchell13} --- and presented here in greater detail. This certificational
chosen-plaintext-based forgery attack limits the security of any scheme using the design shown in
Figure~\ref{fig:mode} (including IOBC and EPBC), regardless of length limits for plaintexts and the
choice of $g$.

\begin{lemma}  \label{l-certificational}
Suppose that $C'_1, C'_2, \ldots, C'_{t'}$ and $C_1, C_2, \ldots, C_{t}$ are encrypted versions of
the plaintext sequences $P'_1, P'_2, \ldots, P'_{t'}$ and $P_1, P_2, \ldots, P_{t}$, respectively,
using the same key $K$ (although the IVs may be different). Suppose also that $P'_{j}=P_{i}$
and $C'_{j}=C_{i}$ for some $j<t'$ and $i<t$. As previously we refer to the `internal values'
computed during encryption of these two messages as $F'_i$, $F_i$, $G'_i$ and $G_i$.

Then (under reasonable assumptions about the random behaviour of the block cipher) with probability
approximately 0.5 it will hold that $F'_{j-1}=F_{i-1}$, $G'_{j-1}=G_{i-1}$, $F'_{j}=F_{i}$ and $G'_{j}=G_{i}$.
\end{lemma}

\begin{proof}
Let the event $E_\Delta$ be that $\Delta=F'_{j-1}\oplus F_{i-1}$.  Then, under reasonable
assumptions about randomness, $Pr(E_\Delta)$ is $2^{-n}$ for any given $\Delta$.

In the case $E_0$, we immediately have $G'_{j-1}=G_{i-1}$. Also, since $P'_{j}=P_{i}$, it follows
immediately that $F'_{j}=F_{i}$, $G'_{j}=G_{i}$ and $C'_{j}=C_{i}$ with probability 1.

Now consider the event $E_\Delta$ for $\Delta\not=0$, i.e.\ $F'_{j-1}\not= F_{i-1}$. Since
$P'_{j}=P_{i}$ this immediately implies that $G'_{j}\not=G_{i}$, and hence
$F'_{j}\not=F_{i}$. Now, since $C'_{j}=g(G'_{j-1})\oplus F'_{j}$ and $C_{i}=g(G_{i-1})\oplus F_{i}$, we have
\[ C'_{j}=C_{i}~~\mbox{if and only if}~~g(G'_{j-1})\oplus g(G_{i-1})=F'_{j}\oplus F_{i}.\]
Under reasonable assumptions about the random behaviour of the encryption function, this will occur with probability $2^{-n}$.  Hence, as $\Delta$ ranges over its $2^n$ possible values, the expected number of times that $C'_{j}=C_{j}$ will hold is approximately 2, one of which will occur when $F'_{j}=F_{i}$.  The result follows. \qed
\end{proof}

We can now give the following simple result that uses the same notation as Lemma~\ref{l-certificational}. Note that it uses Lemma~\ref{splicing-lemma}, which we already observed holds regardless of the choice of $g$.
\begin{lemma}
Suppose that $C'_1, C'_2, \ldots, C'_{t'}$ and $C_1, C_2, \ldots, C_{t}$ are as in the statement of Lemma~\ref{l-certificational}, and suppose also that $P'_{j}=P_{i}$
and $C'_{j}=C_{i}$ for some $j<t'$ and $i<t$. Then, with probability approximately 0.5, the
constructed ciphertext message
\[C'_1, C'_2, \ldots, C'_{j-1}, C_{i}, C_{i+1}, \ldots, C_t\]
will decrypt to $P'_1, P'_2, \ldots, P'_{j-1}, P_{j}, P_{j+1}, \ldots,
P_{t}$.
\end{lemma}
\begin{proof}
The result follows immediately from Lemma~\ref{splicing-lemma}, putting $u=j$, $v=i$ and observing that:
\[ C_v\oplus g(G'_{u-1})\oplus g(G_{v-1}) = C_i\oplus g(G'_{j-1})\oplus g(G_{i-1})\]
which equals $C_i$ with probability approximately 0.5, since, by Lemma~\ref{l-certificational} we know that $G'_{j-1}=G_{i-1}$ with probability approximately 0.5. \qed
\end{proof}

That is, we can construct a forged message that will pass integrity checks with probability 0.5 if we can find a pair of messages $C'_1, C'_2, \ldots, C'_{t'}$ and $C_1, C_2, \ldots, C_{t}$ for which $P'_{j}=P_{i}$ and $C'_{j}=C_{i}$ for some $j<t'$ and $i<t$. If the attacker can arrange for $2^{n/2}$ messages to be encrypted, all containing the same plaintext block (at a known position in each case), then by the usual `birthday paradox' probabilistic arguments, such a pair is likely to arise. In fact, as observed in \cite{Mitchell13}, the number of required message encryptions can be reduced to somewhat less than $2^{n/2}$ by including many occurrences of the fixed plaintext block in each chosen message.

Of course, this is not likely to be a realistic attack in practice; the importance of the above discussion is that it limits the level of security provided by any scheme using the general construction of Figure~\ref{fig:mode}, regardless of the choice of $g$. In the remainder of this section we consider two attack strategies that work for two different general classes of possible functions $g$.

\subsection{A new attack approach with implications for EPBC}  \label{ss-newattack}

We start by giving a simple generalisation of Theorem~\ref{t-PES-PCBC}.

\begin{theorem}  \label{t-PES-PCBC-generalised}
Suppose the ciphertext $C_1,C_2,\ldots,C_t$ was constructed using a scheme of the type shown in
Figure~\ref{fig:mode}, and that $P_j$ is a plaintext block for some $j$ satisfying $1<j<t$. Suppose
the $(t+2)$-block ciphertext $C'_1,C'_2,\ldots,C'_{t+2}$ is constructed as follows:
\begin{eqnarray*}
C'_i & = & C_i,\:\:\:(1\leq i\leq j), \\
C'_{j+1} & = & P_j\oplus G_j \oplus g(G_j), \\
C'_i & = & C_{i-2},\:\:\:(j+2\leq i\leq t+2).
\end{eqnarray*}
When decrypted to yield $P'_1,P'_2,\ldots,P'_{t+2}$, the value of the final plaintext block
$P'_{t+2}$ will equal $P_t$ for the original (untampered) message, and hence will pass the
integrity check.
\end{theorem}

\begin{remark}
In the case where $g$ is the identity function, as is the case for PES-PCBC, then the above result reduces to Theorem~\ref{t-PES-PCBC}.
\end{remark}

\begin{proof}
We need only examine the decryption of $C'_{j+1}$; the rest of the proof is exactly as in the proof of Theorem~\ref{t-PES-PCBC}. Now:
\begin{eqnarray*}
F'_{j+1} & = & C'_{j+1}\oplus G'_j \\
& = & (P_j\oplus G_j \oplus g(G_j)) \oplus G_j \\
& = & P_j\oplus g(G_j) \\
& = & F_{j-1}.
\end{eqnarray*}
Hence $G'_{j+1} = d_K(F'_{j+1}) = d_K(F_{j-1}) = G_{j-1}$. Finally, we have
\[ P'_{j+1} = G'_{j+1}\oplus F'_j = G_{j-1}\oplus F_j = C_j\oplus G_{j-1}\oplus g(G_{j-1}),\]
although the precise value of $P'_{j+1}$ is unimportant. The result follows trivially. \qed
\end{proof}

Of course, the degree to which Theorem~\ref{t-PES-PCBC-generalised} is likely to enable a forgery attack depends very much on the properties of the function $g$. However, we have the following simple result for the function $g$ used in EPBC.

\begin{lemma}  \label{l-g-EPBC}
Suppose $g$ is as defined for EPBC, and (using the notation of Lemma~\ref{g_output}) suppose also that $g(L||R)=L'||R'$. Then:
\begin{enumerate}
\item[(i)] For all inputs $L||R$, we have $L\oplus L'=R\oplus R'$;
\item[(ii)] If $L||R$ is chosen at random, then each bit of $L\oplus L'$ is equal to 1 with probability 0.25.
\end{enumerate}
\end{lemma}

\begin{proof}
Suppose $\ell$ is a bit in $L$, and $r$ is the corresponding bit in $R$.  Suppose also that ($\ell'$,$r'$) are the bits in the same positions in $L'||R'$.  Then
\[ \ell' = (\ell \lor r)\oplus \ell = \lnot\ell \land r. \]
Also
\[ r' = (\ell \land r)\oplus r  = \lnot\ell \land r. \]
Claim (i) follows immediately, and claim (ii) follows from observing that $\lnot\ell \land r=1$ if and only if $r=1$ and $\ell=0$. \qed
\end{proof}

From Theorem~\ref{t-PES-PCBC-generalised}, we can construct a possible forgery by guessing the value of $G_j \oplus g(G_j)$.  If $g$ is as defined for EPBC, then, from Lemma~\ref{l-g-EPBC}(i), we simply need to guess the first $n/2$ bits of $G_j \oplus g(G_j)$.  Moreover, if we restrict our guesses for this `first half' to $n/2$-bit strings containing at most $n/8$ ones (where we assume that $n$ is a multiple of 8), then from Lemma~\ref{l-g-EPBC}(ii) we will have a better than evens chance of making a correct guess.  The number of such strings is simply:
\[ \sum_{i=0}^{n/8} {{n/2} \choose i}. \]
There are many ways of estimating this sum, but the following well known result is helpful.
\begin{lemma}  \label{l-binco}
Suppose $m\geq 1$ and $0\leq k< m/2$.  Then
\[ \sum_{i=0}^{k} {{m} \choose i} < {m \choose k}\frac{m-k+1}{m-2k+1}. \]
\end{lemma}
If $n=64$ or $n=128$ then the above sum is $1.50\times 10^7\simeq 2^{23.8}$ or $7.13\times
10^{14}\simeq 2^{49.3}$, respectively. That is, for $n=64$, after $2^{23.8}$ trials, there is a
good chance one fake message will pass the integrity check, and for $n=128$, $2^{49.3}$ trials will
suffice.

Of course, these are large numbers, but they are significantly less than the certificational attack
with complexity $2^{n/2}$ described in Section~\ref{ss-EPBC-chosen}.

\subsection{Issues with the use of Initialisation Vectors (IVs)}  \label{ss-IV-attack}

We next show how to construct a forgery in any scheme of the type shown in Figure~\ref{fig:mode} if
the IVs (i.e.\ $F_0$ and $G_0$) are not different for every encrypted message and $g$ is linear.

Before discussing the attack we briefly recap what the authors of the three schemes considered here
say about the choice of IVs.
\begin{itemize}
\item In the paper introducing PES-PCBC, \cite{Zuquete96}, there is no mention of how $F_0$ and $G_0$ are chosen --- indeed, the need for them to be selected does not even appear to be mentioned. However, in the subsequent paper introducing EPBC \cite{Zuquete97}, which also points out an attack on PES-PCBC, it is stated that the `initial values of $F_{i-1}$ and $G_{i-1}$ are distinct, secret initialisation vectors'.
\item In the paper introducing EPBC, exactly the same statement, i.e.\ that the `initial values of $F_{i-1}$ and $G_{i-1}$ are distinct, secret initialisation vectors' is made twice, with no further guidance.
\item In the IOBC paper \cite{Recacha96}, the issue is discussed in slightly greater detail. It is stated that it `is a design requirement for IOBC that the initialising vectors \ldots shall be changed for each encrypted message'.
\end{itemize}

Of course, changing the values of $F_0$ and $G_0$ for each encrypted message, as required for IOBC, is clearly good practice.  Indeed, if the same values are used to encrypt two messages which contain the same initial plaintext block, then the resulting ciphertexts will share the same ciphertext block. That is, the mode would leak information about plaintext, which is clearly a highly undesirable property for any mode of operation intended to provide confidentiality. Nonetheless, the designers of EPBC and PES-PCBC did not impose any requirement for the values to be changed.

We can state the following result, which is essentially a special case of Lemma~\ref{splicing-lemma}, and applies to all modes adhering to the design of Figure~\ref{fig:mode} and for which $g$ is linear (as is the case for IOBC and, trivially, PES-PCBC). In such a case, as observed in \cite{Mitchell13}, the distributivity property $g(X\oplus Y)=g(X)\oplus g(Y)$ for any $X$ and $Y$ holds.

\begin{lemma}  \label{generalised-splicing-lemma}
Suppose $C_1, C_2, \ldots, C_t$ and $C'_1, C'_2, \ldots, C'_{t'}$ are encrypted
versions of the plaintext sequences $P_1, P_2, \ldots, P_t$ and $P'_1, P'_2, \ldots, P'_{t'}$,
respectively, where the method of encryption is as shown in Figure~\ref{fig:mode}.  Suppose also that the two messages are encrypted using identical IVs, i.e.\ $F_0=F'_0$ and $G_0=G'_0$.

If the ciphertext
\[ C^*_1, C^*_2, \ldots, C^*_{t} = C'_1, C'_{2}, C_3\oplus g(C_1\oplus C'_1\oplus P_2\oplus P'_2), C_{4}, \ldots, C_t \]
is submitted for IOBC decryption (where $1<u$ and $1<v<t$, and $G_{v-1}$ and $G'_{u-1}$ are values
computed during the encryption of the respective sequences of blocks), then the resulting sequence
of plaintext blocks $P^*_1, P^*_2, \ldots, P^*_{t}$ will be equal to
\[ P'_1, P'_2, P_3 \oplus F'_{2} \oplus F_{2}, P_{4}, P_{5}, \ldots, P_t.
\]
\end{lemma}

\begin{proof}
First observe that, by definition, we know that
\[ G_2\oplus g(G_0) = P_2 \oplus C_1,~~\mbox{and}~~G'_2\oplus g(G'_0) = P'_2 \oplus C'_1. \]
Hence, since we assume that $G_0=G'_0$, we immediately have:
\[ G_2\oplus G'_2 = C_1\oplus C'_1\oplus P_2\oplus P'_2, \]
and thus (using the distributive property of $g$):
\[ g(G_2) \oplus g(G'_2) = g(C_1\oplus C'_1\oplus P_2\oplus P'_2). \]
The result follows from Lemma~\ref{splicing-lemma}, setting $u=v=3$.
\qed
\end{proof}

Thus if two ciphertexts are encrypted using the same IV, and a single plaintext block is known for
each message, then a forgery can be constructed. Observe that this attack can be extended using
Lemma~\ref{induction}{(ii)}.

It may well be the case that other forgery strategies can be devised building on the distributive property when $g$ is linear, but we do not explore this further here.

\section{Other related modes}  \label{s-othermodes}

We conclude this discussion of modes by briefly mentioning two further modes `from the same stable'.

There was a gap of some 16 years before the first of these two additional schemes was made public --- IOC (short for \emph{Input and Output Chaining}) was made public by Recacha in 2013
    \cite{Recacha13}. IOC is clearly closely related to IOBC and EPBC, but was designed to
    avoid the known issues with these schemes. IOC was made public at a time when there was a
    renewed interest in the area, at least partly prompted by a NIST initiative on lightweight
    cryptography\footnote{As discussed at
    \url{https://csrc.nist.gov/Projects/Lightweight-Cryptography}, NIST began investigating
    cryptography for constrained environments in 2013, and one of the goals was to find
    lightweight methods for authenticated encryption.}. A slightly revised version was made
    public in early 2014, \cite{Recacha14a}. Cryptanalyses of both versions of the scheme first
    appeared in 2014 \cite{Bottinelli14}.

The second additional scheme, known as ++AE, is a further evolution of the previous schemes, again
    designed with the intention of addressing the known issues. Like IOC, this scheme exists in
    two slightly different versions, v1.0 \cite{Recacha14b} and v1.1 \cite{Recacha14}, both
    promulgated in 2014. Both versions were cryptanalysed in a pair of papers published in 2016
    and 2018 \cite{AlMahri18,AlMahri16}.

\section{Summary and conclusions}  \label{s-conclusions}

In this paper we have re-examined the security of three closely related block cipher modes of
operation designed to provide authenticated encryption, namely PES-PCBC, IOBC and EPBC\@. Whilst
cryptanalysis of all three schemes has previously been published, the attack on EPBC has
subsequently been shown to be incorrect and hence until now no effective attack was known against this mode.

In this paper we have both elaborated on and enhanced existing cryptanalysis, and we have also
demonstrated new attacks which show that none of the three schemes can be considered secure. The main findings of the paper are as follows.
\begin{itemize}
\item There exists a forgery attack on any scheme of the type shown in Figure~\ref{fig:mode} which requires of the order of $2^{n/2}$ chosen plaintexts --- see Section~\ref{ss-EPBC-chosen}.
\item It was already known (see \cite{Zuquete97}) that simple forgeries against PES-PCBC could be devised requiring a single ciphertext message and two known plaintext blocks for this ciphertext --- a variant of this attack was described (see Section~\ref{ss-PES-PCBC-cryptanalysis} requiring only a single known plaintext block. Another simple forgery attack against PES-PCBC was described in Section~\ref{ss-IV-attack}, which is realisable if IVs are ever re-used.
\item Di et al.\ \cite{Di15} showed that the only known attack on EPBC mode did not work. However, in Section~\ref{ss-newattack} we described a new attack strategy which yields a successful forgery with high probability with significantly fewer than $2^{n/2}$ trials (e.g.\ $2^{23.8}$ trials for $n=64$).
\item Forgery attacks on IOBC were already known (see \cite{Mitchell13}). A further simple forgery attack was described in Section~\ref{ss-IV-attack}, which is realisable if IVs are ever re-used (although it is important to note that IV reuse is specifically prohibited in \cite{Recacha96}).
\end{itemize}

Existing cryptanalysis, when combined with the new attacks described in this paper, suggests very strongly than none of the three modes considered in this paper are sufficiently robust against forgery attacks to be used in practice.

\section*{Dedication}

This paper is dedicated to the memory of Ed Dawson.


\end{document}